\documentclass[10pt]{article}
\pdfoutput=1

\def\inst{\end{tabular}\\\begin{tabular}{c}}

\usepackage{kpfonts}
\usepackage[usenames,dvipsnames]{color}
\usepackage{hyperref}
\hypersetup{
  colorlinks,
  citecolor=Red,
  linkcolor=Blue,
  urlcolor=Blue}

\usepackage[margin=1.25in]{geometry}

\usepackage{amsthm,amsmath,amsfonts,amssymb}

\newtheorem{theorem}{Theorem}
\newtheorem{lemma}[theorem]{Lemma}
\newtheorem{proposition}[theorem]{Proposition}

\theoremstyle{definition}

\newtheorem{example}[theorem]{Example}

\usepackage{parskip}
\begingroup
    \makeatletter
    \@for\theoremstyle:=definition,remark,plain\do{%
        \expandafter\g@addto@macro\csname th@\theoremstyle\endcsname{%
            \addtolength\thm@preskip\parskip
            }%
        }

\usepackage[T1]{fontenc}

\usepackage{subfig}

\usepackage{tikz}
\usetikzlibrary{decorations.pathreplacing}

\usepackage{mathrsfs}

\renewcommand{\le}{\leqslant}
\renewcommand{\ge}{\geqslant}

\newcommand{\abs}[1]{\lvert#1\rvert}

\newcommand{\eps}{\varepsilon}

\newcommand{\R}{\mathbb R}
\DeclareMathOperator*{\E}{\mathbb{E}}

\newcommand{\states}{\mathbb S}
\newcommand{\statesstar}{\hat{\states}}

\DeclareMathOperator*{\argmax}{arg\,max}

\usepackage{xspace}

\newcommand{\OPT}{\textsf{OPT}\xspace}
\newcommand{\ADAPT}{\textsf{ADAPT}\xspace}
\newcommand{\NONADAPT}{\textsf{NON\-ADAPT}\xspace}
\newcommand{\AGREEDY}{\textsf{AGREEDY}\xspace}
\newcommand{\NBR}{\textsf{NBR}\xspace}
\newcommand{\OFFLINE}{\textsf{OFF\-LINE}\xspace}

\title{Discrete Stochastic Submodular Maximization:\\Adaptive vs. Non-Adaptive vs. Offline}
\author{Lisa Hellerstein \and Devorah Kletenik \and Patrick Lin
\inst
Department of Computer Science \& Engineering\\
Polytechnic School of Engineering, 
New York University\\
\texttt%
{\{lisa.hellerstein,dkletenik,patrick.lin\}@nyu.edu}
}

\date{}

\begin{document}
\maketitle

\begin{abstract}
We consider the problem of stochastic monotone submodular function maximization, subject to constraints.
We give results on adaptivity gaps, and on the gap between the optimal offline and online solutions. 
We present a procedure that transforms a decision tree (adaptive algorithm) into a non-adaptive chain. We prove that this chain achieves at least $\tau$ times the utility of the decision tree, over a product distribution and binary state space, where $\tau=\min_{i,j} \Pr[x_i=j]$. This proves an adaptivity gap of $\frac{1}{\tau}$ (which is $2$ in the case of a uniform distribution) for the problem of stochastic monotone submodular maximization subject to state-independent constraints. 
For a cardinality constraint, 
we prove that a simple adaptive greedy algorithm achieves an approximation factor of $(1-\frac{1}{e^\tau})$ with respect to the optimal offline solution; previously, it has been proven that the algorithm achieves an approximation factor of $(1-\frac{1}{e})$ with respect to the optimal adaptive online solution. Finally, we show that there exists a non-adaptive solution for the stochastic max coverage problem that is within a factor $(1-\frac{1}{e})$ of the optimal adaptive solution and within a factor of $\tau(1-\frac{1}{e})$ of the optimal offline solution.

\end{abstract}

\section{Introduction}
We consider
stochastic submodular function maximization, subject to constraints.
This problem is
motivated by problems in application areas such as machine learning, 
social networks, and recommendation systems.

In traditional (non-stochastic) submodular function maximization, 
the goal is to find a subset of ``items'' with maximum utility, as 
measured by a submodular utility function assigning a real value
to each possible subset of items.
In {\em stochastic} submodular function maximization,
items have {\em states}.  For example, if each item is a sensor,
the item might be either working or broken.
The utility of a subset of items
depends not only on which items are in the subset, but also on their states. 
The state of each item is initially unknown, and can only be determined by
performing a ``test'' on the item.

Algorithms for stochastic submodular maximization work in an on-line setting,
sequentially choosing which item to test next.
The choice can be adaptive, depending on the outcomes of previous tests.
The state of each item is an independent random variable.
The goal is to maximize the expected utility of the tested items.
Previous work has sought to 
determine the adaptivity gap, which is the ratio between the optimal adaptive and non-adaptive
solutions.  In this paper we present new adaptivity gap results
for discrete monotone submodular functions.
We also consider another type of gap that has not been previously explored
in the context of stochastic submodular maximization: the ratio between the optimal offline solution and the optimal
adaptive solution.

Our main result is an adaptivity gap of 2 for all {\em state-independent} constraints,
when the state set is binary and the item state distribution is uniform.
More generally, for arbitrary product distributions, we prove an adaptivity gap of $\frac{1}{\tau}$.
Here $\tau$ is the minimum value of $p_{i,j}$, where $p_{i,j}$ is the 
probability that item $i$ is in state $j$.
We say that a constraint is state-independent if the restriction on the items tested
does not depend on their states. 
(A constraint requiring testing to stop when an item is found
to be in state 1 is not state-independent.)
A standard knapsack constraint is
state-independent, and this is the first adaptivity gap
for knapsack constraints.
We prove the gap using
a simple, bottom-up procedure that transforms a decision tree (adaptive algorithm) into a single
non-adaptive chain corresponding to a root-leaf path in the tree.


Asadpour and Nazerzadeh 
previously showed an adaptivity gap of $\frac{e}{e-1}$ for a matroid constraint,
using a stronger monotonicity condition than the one we use here (their results also apply to continuous states)~\cite{asadpour2014maximizing}. 
For a cardinality constraint,
we show that the simple adaptive greedy algorithm gives a
$(1-\frac{1}{e^\tau})$-approximation with respect to the optimal \emph{offline} solution, and
that a dependence on $\tau$ in the approximation factor is necessary.

Finally, we consider the discrete stochastic version of the maximum coverage problem, which is
a special case of submodular maximization subject to a cardinality constraint.
We modify an approximation algorithm for the deterministic version
of this problem, due to Ageev and Sviridenko~\cite{ageev1999approximation,ageev2004pipage}, to 
prove that the optimal non-adaptive solution for this problem is within a
factor of ${1-\frac{1}{e}}$ of the optimal adaptive solution.
We also show that the optimal non-adaptive solution for this problem is within a
factor of ${\tau(1-\frac{1}{e})}$ with respect to the optimal offline solution.

\section{Preliminaries and Definitions}
\label{sec:prelims}

Let $\states=\{0,\ldots,\ell-1\}$, and $\statesstar=\states\cup\{*\}$. A partial assignment is a vector
$b \in \statesstar^n$. Hence $b$ can be viewed as an assignment to variables $x_1 ,\ldots , x_n$.
We will use these partial assignments to represent the outcomes of tests 
giving the states of $n$ items, where each item can be in one of $\ell$ states.
We write $b_i=s$ to indicate that item $i$ has been tested and found to be in state $s$,
and $b_i = *$ to indicate that the state of item $i$ is unknown. We assume that the states of
different items are independent.

If $b',b\in\statesstar^n$ and $b'_i=b_i$ for all $b_i\ne *$, then we call $b'$ an \emph{extension} of $b$,
which we will write as $b' \succ b$.
We will use $b_{j\gets s}$ to denote the extension of $b$ setting the $j$-th bit of $b$ to $s$.

As is standard in the literature, given a set $N=\{1,\ldots,n\}$, we say a function $g:2^N\to\R_{\ge 0}$
is a \emph{utility function}. We will use the notation $g_S(j)$ to denote $g(S\cup\{j\})-g(S)$.

We extend the notion of a utility function to the stochastic setting, wherein we have
$g:\statesstar^n\to\R_{\ge 0}$ defined on partial assignments. 
In this case, we will write
$g(S,b):=g(b')$ where $b'$ is a partial assignment consistent with $b$ on all entries $i$ where 
$i\in S$, and $b_i=*$ whenever $i\not\in S$. The notation $g_{S,b}(j)$ will
denote $g(S\cup\{j\},b)-g(S,b)$. If $S'$ is a set, then $g_S(S')$ will mean $g(S\cup S')-g(S)$.

Utility function $g:\statesstar^n\to\R_{\ge 0}$ is called \emph{submodular} if
$g(b_{i\gets s})-g(b)\ge g(b'_{i\gets s})-g(b')$ when $b' \succ b$, $b_i'=b_i=*$,
and $s\in\states$.
We say $g$ is \emph{monotone} if $g(b_{i\gets s})\ge g(b)$ when $b_i=*$.
That is, testing a bit can only increase the utility. 

We will work with product distributions over the vectors $\states^n$: for $i\in N$, $j\in\states$, we use $p_{i,j}$ to mean the probability of the
$i$-th coordinate being $j$ (so that for each $i$, $\sum_j p_{i,j}=1$).
Many of our results are with respect to $\tau = \min_{i, j} p_{i,j}$.
We use $\E[g_{S,b}(j)]$ to denote the expected increase in utility from testing the $j$-th bit.

We define the \emph{Stochastic Submodular Maximization} problem as the problem of 
maximizing a monotone submodular function,
in the stochastic setting with a discrete state space, subject to
one or more constraints.
More specifically, in this problem,
we are given as input a monotone submodular
$g:\statesstar^n\to\R_{\ge 0}$, the constraints, and the parameters of a product distribution over $\statesstar^n$.

Solving a Stochastic Submodular Maximization problem entails finding an \emph{adaptive solution}
that builds a set $Q\subseteq N$ item by item, testing each item after selecting it, that maximizes $\E[g(Q,b)]$ subject to the constraints. 
This is effectively a decision tree, whose nodes are labeled with $j\in N$, and we branch
depending on the outcome of $b_j$. A \emph{chain} is a balanced tree such that the labels are the same for all nodes at the same level.
A chain is a \emph{non-adaptive procedure}. We also consider the so-called \emph{offline solution}, which knows \emph{a priori}
the outcomes of the random bits, and takes the optimal set with respect to said outcome.

A knapsack constraint has the form
$\sum_{j\in Q} c_j \le B$ where each $c_j \ge 0$, and $B \ge 0$.
The $c_j$ are called costs, and $B$ is a budget.
A special case is when $c_j = 1$ for all $j$ and $B$ is an integer;
this is a cardinality constraint.
(A cardinality constraint is also a special case of a matroid constraint.)

The \emph{Stochastic Max Coverage} problem is a special case of the 
Stochastic Submodular Maximization problem with a cardinality constraint.
The utility function $g:\statesstar\to\R_{\ge0}$ in the Stochastic Max Coverage problem
is defined as follows: let $E=\{e_1,\ldots,e_m\}$ be a ground set of elements, and for $i\in N$ and $r\in\states$, $S_{i,r}\subseteq E$. For convenience of notation, we will also write $S_{i,a}$ to mean $S_{i,a_i}$ where $a\in\states^n$. Then $g(S,b)=\abs{\bigcup_{i\in S} S_{i,b}}$. If $b$ is a partial assignment, then we say $e_j$ is \emph{covered} with respect to $b$ if $e_j\in\bigcup_{i:b_i\ne *} S_{i,b}$. This function $g$ is clearly submodular and monotone.

The expected values of the optimal adaptive, non-adaptive, and offline solutions will be denoted by
\ADAPT, \NONADAPT, and \OFFLINE. We are interested in the
\emph{adaptivity gap} $\frac{\ADAPT}{\NONADAPT}$, as well as the ratios $\frac{\OFFLINE}{\ADAPT}$ and $\frac{\OFFLINE}{\NONADAPT}$.

\section{Related Work}
\label{sec:related}

The submodular maximization problems studied in this paper were all initially studied in the
deterministic setting.
Feige showed that for all of these problems, under the assumption of $\mathrm{P}\ne\mathrm{NP}$, no polynomial time algorithm can achieve an approximation factor better than $(1-\frac{1}{e})$~\cite{feige1998threshold}.
For the problem of maximizing a monotone submodular function subject to a cardinality constraint,
Nemhauser et al. showed that the natural greedy algorithm achieves an approximation factor of $(1-\frac{1}{e})$~\cite{nemhauser1978best}.
For the problem with a knapsack constraint,
Sviridenko subsequently showed that an algorithm of Khuller et al. also achieves an approximation factor of $(1-\frac{1}{e})$
~\cite{khuller1999budgeted,sviridenko2004note}. The results of Golovin and Krause achieve the same approximation factor for a cardinality constraint in the stochastic setting~\cite{golovin2011adaptive}.

As mentioned earlier, Asadpour and Nazerzadeh showed an adaptivity
gap of $\frac{e}{e-1}$ for Stochastic Submodular Maximization with a matroid constraint,
using a stronger definition of monotonicity.
In their definition of monotonicity,
for any partial assignment $b$ and $s\in\states$, they require that
$g(b_{i \gets s})\ge g(b)$ if either $b_i=*$ or $s \ge b_i$, whereas we only require that
$g(b_{i \gets s}) \ge g(b)$ if $b_i = *$.
Their proof is based on
Poisson clocks and pipage rounding~\cite{asadpour2014maximizing}, and applies to continuous state spaces.
Our adaptivity gap results do not apply to continuous state spaces, but our proofs are
combinatorial.

Chan and Farias studied the related Stochastic Depletion problem and gave a $\frac{1}{2}$-approximation for the problem with respect to
what they call the offline solution in their model~\cite{chan2009stochastic}; in our model, their algorithm translates to a $\frac{1}{2}$-approximation with respect to \ADAPT for Stochastic Submodular Maximization with a cardinality constraint.

\begin{table}[ht]
\centering
{\renewcommand{\arraystretch}{1.25}
\begin{tabular}{|l|c|c|c|} \hline
	&	Knapsack Constraint
	&	Cardinality Constraint
	&	Max Coverage
	\\\hline
Deterministic
	&	$(1-\frac{1}{e}) \OPT$~\cite{sviridenko2004note}
	&	$(1-\frac{1}{e})\OPT$~\cite{nemhauser1978analysis}
	&	$(1-\frac{1}{e}) \OPT$~\cite{khuller1999budgeted}
	\\\hline
\parbox[c]{2cm}{\raggedright Stochastic: Adaptive}
	&	OPEN
	&	\begin{tabular}{c}
		$(1-\frac{1}{e})\ADAPT$~\cite{golovin2011adaptive,asadpour2014maximizing} \\ 
		$\mathbf{(1-\frac{1}{e^\tau})}\textbf{\OFFLINE}$ \\
		\end{tabular}
	&	\begin{tabular}{c}
		$(1-\frac{1}{e})\ADAPT$~\cite{golovin2011adaptive,asadpour2014maximizing} \\
		$\mathbf{(1-\frac{1}{e^\tau})}\textbf{\OFFLINE}$
		\end{tabular}
	\\\hline
\parbox[c]{2cm}{\raggedright Stochastic: Non-Adaptive}
	&	$\boldsymbol{\tau}\textbf{\ADAPT}$
	&	\begin{tabular}{c}
		$(1-\frac{1}{e})\ADAPT^\dagger$~\cite{asadpour2014maximizing}\\ 
		$\boldsymbol{\tau}\textbf{\ADAPT}$
		\end{tabular}
	&	\begin{tabular}{l}
		$(1-\frac{1}{e})\ADAPT^\dagger$~\cite{asadpour2014maximizing}\\ 
		$\mathbf{(1-\frac{1}{e})}\textbf{\ADAPT}$\\
		$\boldsymbol{\tau}\mathbf{(1-\frac{1}{e})}\textbf{\OFFLINE}$
		\end{tabular}
	\\\hline
\end{tabular}
}
\caption{Bounds for Stochastic Submodular Maximization}
\label{tbl:bounds}
\end{table}

Table~\ref{tbl:bounds} summarizes approximation bounds for stochastic discrete monotone submodular function maximization
with a knapsack constraint, a cardinality constraint, and for max-coverage.
Results from this paper are in bold. We denote with a $\dagger$ bounds relying on the stronger definition of monotonicity of~\cite{asadpour2014maximizing}.

The entries in the first row give the best bounds known for polynomial-time algorithms solving the deterministic versions of the problems, assuming oracle access to the utility function $g$.  The bounds are given in terms of \OPT, the optimal solution to the deterministic problem. As noted above, these are the best bounds possible, assuming $\mathrm{P}\ne\mathrm{NP}$~\cite{feige1998threshold}.

The entries in the second row refer to the best nontrivial bounds achieved by polynomial-time algorithms for the stochastic versions of the problems, assuming polynomial-time access to $g$.   Both bounds are achieved by the Adaptive Greedy algorithm (described in Section~\ref{sec:cardinalityADAPTOFFLINE}). We note that Golovin and Krause give a randomized version of the algorithm for the knapsack constraint achieving an approximation factor of $(1-\frac{1}{e})$, but with the relaxation that the budget only needs to be met in expectation~\cite{golovin2011adaptive}.
The last row refers to the bounds achieved by the respective best possible non-adaptive solutions for the problems in the stochastic setting (irrespective of running time).

\section{An Adaptivity Gap for State-Independent Constraints}
\label{sec:conquest}

In this section we present an adaptivity gap for Stochastic Submodular Maximization with state-independent constraints. We use a technique that takes a decision tree and outputs a root-leaf path by collapsing the tree bottom up in a greedy manner; at each step, one child chain of a node replaces the other, leaving a single longer chain.

We show that under a product distribution over $\{0,1\}^n$, this gives a non-adaptive procedure that is a $\tau$-approximation of expected utility of the original tree, where $\tau=\min_{i,j} p_{i,j}$. This gives a bound on the adaptivity gap for the problem: $\frac{\ADAPT}{\NONADAPT}\le \frac{1}{\tau}$ for binary states.

\begin{theorem}
For binary states,
the Stochastic Submodular Maximization problem with state-independent constraints
has an adaptivity gap of at most $\frac{1}{\tau}$.
\end{theorem}

\begin{proof}
Let $T$ be a decision tree corresponding to a solution to an instance of Stochastic
Submodular Maximization with state-independent constraints, and binary states.
We show that if $T$ achieves expected utility $U$, then there exists a chain, 
corresponding to a root-leaf path in $T$,
that achieves expected utility $\tau\cdot U$.
Since the constraints are state-independent, this root-leaf path must obey
the constraints, and the theorem follows.

We use a recursive procedure to turn $T$ into a chain. At any intermediate step, we have a subtree consisting of a parent node whose child subtrees are chains. We show that when performing the procedure on this subtree, the loss incurred in expected utility is at most $1-\tau$ times the expected utility contributed by the parent node. Since the expected utility of a decision tree is a weighted sum of the expected utility contributed by its nodes, the total loss from the entire procedure is at most $1-\tau$ times the expected utility of the whole tree.

Hence without loss of generality, suppose $T$ is a tree with a root node labeled $x_i$ for some $i\in N$ whose two child subtrees are chains, as shown in Figure~\ref{fig:T}.
For convenience of notation, assume the nodes on the left child chain are labeled $x_{l_1}$ through $x_{l_c}$, and the nodes on the right child chain are labeled $x_{r_1}$ through $x_{r_d}$.

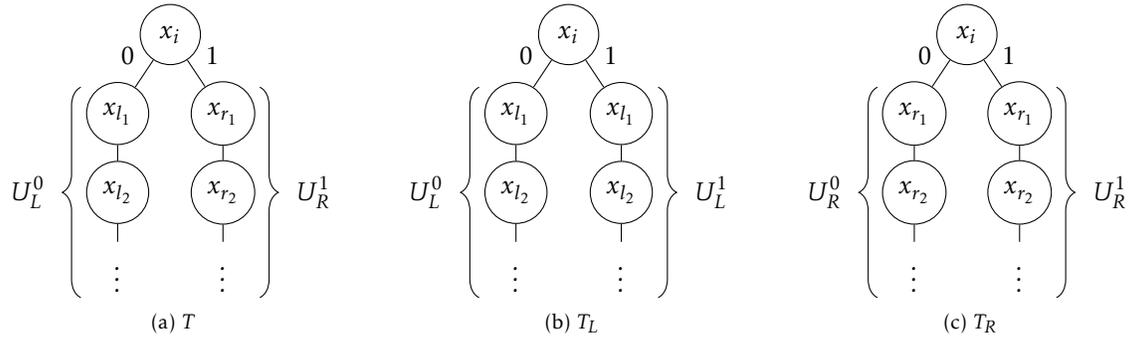
\begin{figure}[ht]
\centering
\subfloat[$T$\label{fig:T}]{%
\begin{tikzpicture}[scale=0.7]
\node[draw,circle,minimum width=0.8cm] (x0) at (0,0) {$x_i$};
\node[draw,circle,minimum width=0.8cm] (x1L) at (-1,-1.5) {$x_{l_1}$};
\node[draw,circle,minimum width=0.8cm] (x2L) at (-1,-3) {$x_{l_2}$};
\node (x3L) at (-1,-4.5) {$\vdots$};
\node[draw,circle,minimum width=0.8cm] (x1R) at (1,-1.5) {$x_{r_1}$};
\node[draw,circle,minimum width=0.8cm] (x2R) at (1,-3) {$x_{r_2}$};
\node (x3R) at (1,-4.5) {$\vdots$};
\draw (x0) to node[above left]{$0$} (x1L);
\draw (x0) to node[above right]{$1$} (x1R);
\draw (x1L) to (x2L);
\draw (x2L) to (x3L);
\draw (x1R) to (x2R);
\draw (x2R) to (x3R);
\draw[decorate,decoration={brace,mirror,amplitude=7pt}] (-1.7,-1) to node[left=1em] {$U_L^0$} (-1.7,-5);
\draw[decorate,decoration={brace,amplitude=7pt}] (1.7,-1) to node[right=1em] {$U_R^1$} (1.7,-5);
\end{tikzpicture}
}
\qquad
\subfloat[$T_L$\label{fig:Tleft}]{%
\begin{tikzpicture}[scale=0.7]
\node[draw,circle,minimum width=0.8cm] (x0) at (0,0) {$x_i$};
\node[draw,circle,minimum width=0.8cm] (x1L) at (-1,-1.5) {$x_{l_1}$};
\node[draw,circle,minimum width=0.8cm] (x2L) at (-1,-3) {$x_{l_2}$};
\node (x3L) at (-1,-4.5) {$\vdots$};
\node[draw,circle,minimum width=0.8cm] (x1R) at (1,-1.5) {$x_{l_1}$};
\node[draw,circle,minimum width=0.8cm] (x2R) at (1,-3) {$x_{l_2}$};
\node (x3R) at (1,-4.5) {$\vdots$};
\draw (x0) to node[above left]{$0$} (x1L);
\draw (x0) to node[above right]{$1$} (x1R);
\draw (x1L) to (x2L);
\draw (x2L) to (x3L);
\draw (x1R) to (x2R);
\draw (x2R) to (x3R);
\draw[decorate,decoration={brace,mirror,amplitude=7pt}] (-1.7,-1) to node[left=1em] {$U_L^0$} (-1.7,-5);
\draw[decorate,decoration={brace,amplitude=7pt}] (1.7,-1) to node[right=1em] {$U_L^1$} (1.7,-5);
\end{tikzpicture}
}
\qquad
\subfloat[$T_R$\label{fig:Tright}]{%
\begin{tikzpicture}[scale=0.7]
\node[draw,circle,minimum width=0.8cm] (x0) at (0,0) {$x_i$};
\node[draw,circle,minimum width=0.8cm] (x1L) at (-1,-1.5) {$x_{r_1}$};
\node[draw,circle,minimum width=0.8cm] (x2L) at (-1,-3) {$x_{r_2}$};
\node (x3L) at (-1,-4.5) {$\vdots$};
\node[draw,circle,minimum width=0.8cm] (x1R) at (1,-1.5) {$x_{r_1}$};
\node[draw,circle,minimum width=0.8cm] (x2R) at (1,-3) {$x_{r_2}$};
\node (x3R) at (1,-4.5) {$\vdots$};
\draw (x0) to node[above left]{$0$} (x1L);
\draw (x0) to node[above right]{$1$} (x1R);
\draw (x1L) to (x2L);
\draw (x2L) to (x3L);
\draw (x1R) to (x2R);
\draw (x2R) to (x3R);
\draw[decorate,decoration={brace,mirror,amplitude=7pt}] (-1.7,-1) to node[left=1em] {$U_R^0$} (-1.7,-5);
\draw[decorate,decoration={brace,amplitude=7pt}] (1.7,-1) to node[right=1em] {$U_R^1$} (1.7,-5);
\end{tikzpicture}
}
\caption{$T$ and its variants}
\end{figure}

Let $L =\{l_m\}_{m=1}^c$ and $R=\{r_m\}_{m=1}^d$.
For $D\in\{L,R\}$ and fixed assignment $b$, monotonicity gives $g_{\varnothing,b}(\{i\}\cup D)\ge g_{\varnothing,b}(D)$,
and submodularity gives $g_{\varnothing,b}(D)\ge g_{\{i\},b}(D)$. Thus, in expectation,
\begin{equation}\label{eq:u0ineq0}
\E[g_{\varnothing,b}(\{i\}\cup L)\mid b_i=1]\ge\E[g_{\varnothing,b}(L)\mid b_i=1]
=\E[g_{\varnothing,b}(L)\mid b_i=0]\ge\E[g_{\{i\},b}(L)\mid b_i=0]
\end{equation}
and similarly
\begin{equation}\label{eq:u1ineq0}
\E[g_{\varnothing,b}(\{i\}\cup R)\mid b_i=0] 
\ge\E[g_{\{i\},b}(R)\mid b_i=1].
\end{equation}

We introduce the following notation:
Set $U_i^0=\E[g_{\varnothing,b}(i)\mid b_i=0]$, $U_i^1=\E[g_{\varnothing,b}(i)\mid b_i=1]$ (note that $g_{\varnothing,b}(i)$ is
actually constant on all $b$ with the same value for $b_i$), and $U_i=\E[g_{\varnothing,b}(i)]=p_{i,0}U_i^0+p_{i,1}U_i^1$. So $U_i$ is the expected utility contributed by the root node $x_i$.

Set $U_L^0=\E[g_{\{i\},b}(L)\mid b_i=0]$, that is, the expected increase in utility from testing $x_{l_1},\ldots,x_{l_c}$
after testing $x_i$ and getting $b_i=0$. Similarly, we set $U_L^1=\E[g_{\{i\},b}(L)\mid b_i=1]$, $U_R^0=\E[g_{\{i\},b}(R)\mid b_i=0]$
and $U_R^1=\E[g_{\{i\},b}(R)\mid b_i=1]$.

With this notation in mind, we can rewrite, respectively, \eqref{eq:u0ineq0} and \eqref{eq:u1ineq0} as
$U_i^1 \ge U_L^0-U_L^1$ and $U_i^0 \ge U_R^1-U_R^0$,
that is,
\begin{equation}
U_i \ge p_{i,0}(U_R^1-U_R^0) + p_{i,1}(U_L^0-U_L^1). \label{eq:uineq}
\end{equation}

Consider the two variants $T_L$ and $T_R$ of $T$ as follows: $T_L$ is the result of replacing the left child chain of $T$ with the right child chain, i.e, both the left and right child chains are labeled $x_{l_1}$ through $x_{l_d}$; $T_R$ is instead the result of replacing the right child chain with the left. $T_L$ and $T_R$ are shown in Figures~\ref{fig:Tleft} and \ref{fig:Tright}, respectively. Note that they are both fully non-adaptive, and thus can just as easily be represented by pure chains.

Without loss of generality, we will assume
\begin{equation}
p_{i,0}(U_L^0-U_R^0) \le p_{i,1}(U_R^1-U_L^1)
\label{eq:assump}
\end{equation}
and substitute $T_R$ for $T$ by replacing the left child chain of $T$ with its
right child chain (in the symmetric setting, we have $p_{i,0}(U_L^0-U_R^0) \ge p_{i,1}(U_R^1-U_L^1)$ and substitute $T_L$ for $T$). Set
\[ S(b)=\begin{cases} \{i\}\cup L & \text{when $b_i=0$} \\ \{i\}\cup R & \text{when $b_i=1$} \end{cases}. \]

Let $\Delta$ be the expected loss in utility. We want to show that $(1 - \tau)U_i \ge \Delta$. We compute $\Delta$:
\begin{align*}
\Delta & = \E[g(S(b),b)]-\E[g(\{i\}\cup R,b)] \\
		& = [p_{i,0}(U_i^0+U_L^0)+p_{i,1}(U_i^1+U_R^1)] 
		- [p_{i,0}(U_i^0+U_R^0)+p_{i,1}(U_i^1+U_R^1)] \\
		& = p_{i,0}(U_L^0-U_R^0).
\end{align*}

There are two cases: (i) $p_{i,0}\le p_{i,1}$ and (ii) $p_{i,0}> p_{i,1}$.

In case (i) we show that $\frac{1}{p_{i,1}}\Delta\le U_i$, from which it follows that $\Delta \le (1-\tau) U_i$ since $\tau \le p_{i,0} = 1 - p_{i,1}$. We have
\begin{align*}
\frac{1}{p_{i,1}}\Delta &= \left(1+\frac{p_{i,0}}{p_{i,1}}\right)\Delta \\
				&\le \Delta + p_{i,0}(U_R^1-U_L^1) \\
				&= p_{i,0}[(U_L^0-U_R^0)+(U_R^1-U_L^1)] \\
				&= p_{i,0}[(U_L^0-U_L^1)+(U_R^1-U_R^0)] \\
				&\le p_{i,0}(U_L^0-U_L^1)+p_{i,1}(U_R^1-U_R^0) \\
				&\le U_i
\end{align*}
where the first inequality follows by \eqref{eq:assump}, the second inequality since $p_{i,0}\le p_{i,1}$, and the last inequality from \eqref{eq:uineq}.

In case (ii) we show instead that $\frac{1}{p_{i,0}}\Delta\le U_i$, from which again we have $\Delta \le (1-\tau) U_i$ since $\tau \le p_{i,1} = 1 - p_{i,0}$. The computation is similar:
\begin{align*}
\frac{1}{p_{i,0}}\Delta &= \left(1+\frac{p_{i,1}}{p_{i,0}}\right)\Delta \\
				&\le \Delta+p_{i,1}(U_L^0-U_R^0) \\
				&\le p_{i,1}[(U_R^1-U_L^1)+(U_L^0-U_R^0)] \\
				&= p_{i,1}[(U_L^0-U_L^1)+(U_R^1-U_R^0)] \\
				&< p_{i,0}(U_L^0-U_L^1)+p_{i,1}(U_R^1-U_R^0).
\end{align*}
This completes the proof.
\end{proof}


\section{The Gap Between \ADAPT and \OFFLINE}\label{sec:cardinalityADAPTOFFLINE}

In this section, we consider the gap $\frac{\OFFLINE}{\ADAPT}$ for Stochastic Submodular Maximization with a cardinality constraint.

We first define some notation. We use $k$ to denote the number of items allowed by the cardinality constraint.
For consistency, any variant of $b$ will refer to a partial assignment,
and any variant of $a$ will refer to only a full assignment, that is, $a\in\states^n$.
For the sake of clarity, in this section we will explicitly specify the assignments over which we are taking expectations,
except with the shorthand that
when $b$ is a fixed partial assignment, $\E_{a\succ b}[\,\cdot\,]$ will mean $\E_a[\,\cdot\mid a\succ b]$.
If we have a sequence $\{S^t\}_{t=0}^k$, we write $g_t(j)$ in place of $g_{S^t}(j)$.

We use \AGREEDY to denote the expected utility of the \emph{Adaptive Greedy} algorithm of Golovin and Krause~\cite{golovin2011adaptive} (also called \emph{Adaptive Myopic} by Asadpour and Nazerzadeh~\cite{asadpour2014maximizing}), which, starting with $Q^0=\varnothing$, at each step $t$, based on the partial assignment $b^{t-1}$ of bits tested so far, adaptively picks $i_t$ satisfying
\[ i_t=\argmax_{i\in N\setminus Q^{t-1}} \E_{a\succ b^{t-1}}\left[g_{t-1,a}(i)\right], \]
sets $Q^t=Q^{t-1}\cup\{i_t\}$, and tests $b_{i_t}$ to get $b^t$. This implicitly forms a decision tree of depth $k$
that branches based on the outcome of $b_{i_t}$, and outputs $Q^k$.

It is clear that $\AGREEDY\le\ADAPT\le\OFFLINE$. We show:

\begin{theorem}\label{thm:1-1/etau}
$\AGREEDY\ge(1-\frac{1}{e^\tau})\cdot\OFFLINE$ for Stochastic Submodular Maximization with a cardinality constraint.
\end{theorem}

In other words, $\frac{\OFFLINE}{\ADAPT}\le\frac{e^\tau}{e^\tau-1}$. Although this is bound is weak for small $\tau$, we observe that some dependence on $\tau$ is unfortunately unavoidable:

\begin{proposition}\label{thm:adapttau}
\ADAPT cannot achieve an approximation bound better than $\tau$ relative to \OFFLINE.
\end{proposition}

The proof of Proposition~\ref{thm:adapttau} is by example, and is given in the appendix. We will note, however, that $\tau$ and $1-\frac{1}{e^\tau}$ are close: for $0<\tau\le\frac{1}{2}$ the difference between $\tau$ and $1-\frac{1}{e^\tau}$ is at most ${\sim}0.107$, which is achieved at $\tau=\frac{1}{2}$.

We now prove Theorem~\ref{thm:1-1/etau}. We use the following lemma due to Wolsey:

\begin{lemma}[\cite{wolsey1982maximising}]
Let $k$ be a positive integer, and $s>0$, $\rho_1,...,\rho_k\ge0$ be reals. Then
\[ \frac{\sum_{i=1}^k \rho_i}{\min_{t\in\{1,\ldots,k\}}\left(s\rho_t + \sum_{i=1}^{t-1} \rho_i\right)} \ge 1-\left(1-\frac{1}{s}\right)^k\ge 1-\frac{1}{e^{k/s}}. \]
\label{lem:1-1/eTech}
\end{lemma}

Next, let $Q^*_a$ be the optimal offline solution on assignment $a$, $\tau=\min_{i,j} p_{i,j}$, and $(Q^t,b^t)$ be the collection $Q^t$ given by the Adaptive Greedy algorithm at step $t$ corresponding to the partial assignment $b^t$. Then:

\begin{lemma}
For $t=1,2,\ldots,k$,
\[ \E_a\left[g(Q^*_a,a)\right] \le \E_{a}\left[g(Q^{t-1},a)\right]+\frac{k}{\tau}\cdot\E_a\left[g(Q^t,a)-g(Q^{t-1},a)\right] \]
\label{lem:Wolsey3.5OPT}
\end{lemma}

\begin{proof}
Suppose Greedy chooses $i_t$ at step $t$, that is, $Q^t\setminus Q^{t-1}=\{i_t\}$. By definition we have
\[ \E_{a'\succ b^{t-1}}\left[g_{t-1,a'}(i_t)\right] = \sum_{m=1}^\ell p_{i_t,m}\left[g_{t-1,b^{t-1}_{i_t\gets m}}(i_t)\right]. \]

Next, suppose $j\in Q^*_{a'}\setminus Q^{t-1}$ where $a'$ is a full assignment and $a'\succ b^{t-1}$. We can write $\Pr[a']=\prod_i \psi_i$ where $\psi_i\in\{p_{i,0},\ldots,p_{i,\ell-1}\}$. If we then let $a'_{j\gets m}$ be the assignment $a'$ with the $j$-th bit set to $m$, we have

\begin{align*}
g_{t-1,a'}(j)					& \le \frac{1}{\psi_j}\left(\sum_{m=0}^{\ell-1} p_{j,m}\left[g_{t-1,a'_{j\gets m}}(j)\right]\right) 
							\le \frac{1}{\tau} \left(\sum_{m=0}^{\ell-1} p_{i_t,m}\left[g_{t-1,b^{t-1}_{i_t\gets m}}(i_t)\right]\right) 
							= \frac{1}{\tau}\cdot\E_{a'\succ b^{t-1}}\left[g_{t-1,a'}(i_t)\right]
\end{align*}
where the first inequality follows from the fact that $\psi_j=p_{j,m}$ for some $m$, and the second from the definitions of $\tau$ and $i_t$.

Then since $\abs{Q^*_{a'}\setminus Q^{t-1}}\le k$ we get
\begin{equation}\label{eq:k/tau}
\sum_{j\in Q^*_{a'}\setminus Q^{t-1}} g_{t-1,a'}(j) \le \frac{k}{\tau} \cdot \E_{a'\succ b^{t-1}}\left[g_{t,a'}(i_t)\right].
\end{equation}

Next, we have
\begin{align*}
\E_a\left[g(Q^*_a,a)\right] 	& = \E_{b^{t-1}}\left[ \E_{{a'}\succ b^{t-1}} \left[g(Q^*_{a'},{a'})\right]\right] \\
				& \le \E_{b^{t-1}}\left[ \E_{{a'}\succ b^{t-1}} \left[g(Q^*_{a'} \cup Q^{t-1},{a'}) \right]\right] \\
				& \le \E_{b^{t-1}}\left[ \E_{{a'}\succ b^{t-1}} \left[g(Q^{t-1},{a'})+\sum_{j\in Q^*_{a'}\setminus Q^{t-1}} g_{t-1,a'}(j) \right] \right] \\
				& \le \E_{a}\left[g(Q^{t-1},a)\right]+\E_{b^{t-1}}\left[\frac{k}{\tau}\cdot\E_{{a'}\succ b^{t-1}}\left[g_{t-1,a'}(i_t)\right] \right] \\
				& = \E_{a}\left[g(Q^{t-1},a)\right] + \frac{k}{\tau}\cdot\E_a\left[g_{t-1,a}(i_t)\right]
\end{align*}
where the first inequality follows from monotonicity, the second from submodularity, and the third from \eqref{eq:k/tau}.
\end{proof}

In light of Lemma~\ref{lem:1-1/eTech} we see that Theorem~\ref{thm:1-1/etau} follows from Lemma~\ref{lem:Wolsey3.5OPT}:

\begin{proof}[Proof of Theorem~\ref{thm:1-1/etau}]
Let $\rho_i=\E_a[g(Q^i,a)-g(Q^{i-1},a)]$, then clearly we have $\E_a[g(Q^{t-1},a)]=\sum_{i=1}^{t-1} \rho_i$. Since the Adaptive Greedy algorithm outputs $Q^k$,
\begin{align*}
\frac{\AGREEDY}{\OFFLINE}	& = \frac{\sum_{i=1}^k \rho_i}{\E_a[g(Q^*_a,a)]} 
									\ge \frac{\sum_{i=1}^k \rho_i}{\min_{t\in\{1,\ldots,k\}}\left(\frac{k}{\tau}\rho_t+\sum_{i=1}^{t-1} \rho_i\right)} 
									 \ge 1-\frac{1}{e^\tau}
\end{align*}
as desired.
\end{proof}

\section{Gaps for Stochastic Max Coverage}

In this section we consider the special case of Stochastic Max Coverage. We still achieve an adaptivity gap of $\frac{e}{e-1}$, which is tight by an example given by Asadpour and Nazerzadeh~\cite[\S3.1]{asadpour2014maximizing}. Furthermore, we obtain a bound of $\frac{e}{\tau(e-1)}$ for $\frac{\OFFLINE}{\NONADAPT}$.

We use of the following fact, which is easily seen by inspection or by calculus.

\begin{lemma}
For all $x\ge 0$, $1-(1-\frac{x}{k})^k\ge x(1-\frac{1}{e})$.
\end{lemma}

We give a combinatorial proof of the following result, which is adapted from 
work of
Ageev and Sviridenko on the deterministic version of the problem~\cite{ageev1999approximation,ageev2004pipage}.

\begin{theorem}
For the Stochastic Max Coverage problem, there exists a non-adaptive solution that achieves coverage $(1-\frac{1}{e})\ADAPT$.
\label{thm:uniformnbr}
\end{theorem}

\begin{proof}
For simplicity, we give the analysis for the uniform distribution over binary states; the analysis is similar for product distributions over $\ell>2$ states.

We make use of the \emph{neighbor property}, which every adaptive algorithm satisfies by definition: given two assignments $a,a'$ differing only in bit $j$, either $b_j$ is tested for both assignments, or for neither. We denote the expected value of the optimal offline solution satisfying the neighbor property by \NBR. Clearly, $\NONADAPT\le\ADAPT\le\NBR\le\OFFLINE$. We will prove that $\NONADAPT\ge(1-\frac{1}{e})\NBR$, thus achieving the desired adaptivity gap. Let $\mathcal{X}$ 
denote the procedure giving the optimal offline solution satisfying the neighbor property.

For each $i\in N$ and assignment $a$ we assign a variable $x_{i,a}$ such that $x_{i,a}=1$ if and only if $S_{i,a}$ is included in $Q_a$, the subcollection given by $\mathcal{X}$. Correspondingly we assign to each $e_j\in E$ a variable $y_{j,a}$ such that $y_{j,a}=1$ if and only if $j\in\bigcup_{i\in Q_a} S_{i,a}$. Since $\NBR$ denotes the expected number of ground elements by the solution given by $\mathcal{X}$, $\NBR=\sum_j\E_a[y_{j,a}]$.

Consider the following (randomized) algorithm for producing a non-adaptive solution from the solution given by $\mathcal{X}$: randomly pick $k$ sets according to the following probability distribution: pick $i$ with probability $\frac{1}{k}\sum_{a}\Pr[a]\sum_{i=1}^n x_{i,a}=\frac{1}{k2^n}\sum_{a}\sum_{i=1}^n x_{i,a}$.

For each $e_j$ we say $i$ is a \emph{promising cover} for $j$ if either $e_j\in S_{i,0}$ or $e_j\in S_{i,1}$. We divide the promising covers into two categories: $i$ is of type $B$ if $e_j\in S_{i,0}\cap S_{i,1}$ and $i$ is of type $A$ otherwise. We abuse notation slightly and let $A=\frac{1}{2^n}\sum_{i\text{ of type }A}\sum_{a:e_j\in S_{i,a}} x_{i,a}$, and similarly $B=\frac{1}{2^n}\sum_{i\text{ of type }B}\sum_a x_{i,a}$. By definition $A+B\ge\E_a[y_{j,a}]$.

In expectation, the probability that a random $i$ produces a promising cover of type $A$ for $e_j$ is $\frac{2A}{k}$; hence the probability that, in expectation, at least one of the chosen $i$ is a promising cover of type $A$ for $e_j$ is at least $1-(1-\frac{2A}{k})^k\ge2A(1-\frac{1}{e})$. A promising cover of type $A$ covers $e_j$ with probability $\frac{1}{2}$, so $e_j$ is covered by at least one promising cover of type $A$ with probability at least $A(1-\frac{1}{e})$.

Similarly, picking a random set produces a promising cover of type $B$ for $j$ with probability $\frac{B}{k}$, so at least one of the chosen $i$ is a promising cover of type $B$ for $e_j$ is at least $1-(1-\frac{B}{k})^k\ge B(1-\frac{1}{e})$, and a promising cover of type $B$ covers $e_j$ with probability 1, so $e_j$ is covered by at least one promising cover of type $B$ with probability at least $B(1-\frac{1}{e})$.
Since the promising covers of type $A$ and $B$ are disjoint, we conclude that $e_j$ is covered with probability at least $(1-\frac{1}{e})(A+B)\ge(1-\frac{1}{e})\E_a[y_{j,a}]$.

Before finishing the proof, we note that for product distributions, we will actually have many more such categories: for each $i\in N$ we will have a different set of categories. Furthermore, when there are $\ell > 2$ states, there is a larger number of possible ways of covering each $e_j$ and hence a larger number of categories. However, the analysis will be similar, so that it will still the case that $e_j$ is covered with probability $\ge(1-\frac{1}{e})\E_a[y_{j,a}]$.

We return to the proof at hand. By linearity of expectation, the expected number of elements covered by this non-adaptive solution is equal to the sum of the probabilities that each $j$ is covered, in other words,
\begin{align*}
\E[\text{\# elements covered}] &= \sum_{j} \Pr[j\text{ is covered}] 
						\ge \left(1-\frac{1}{e}\right)\sum_{j}\E_a[y_{j,a}]=\left(1-\frac{1}{e}\right) \NBR.
\end{align*}
Since this randomized procedure achieves at least $(1-\frac{1}{e})\NBR$ in expectation, then there must exist a non-adaptive solution that achieves at least $(1-\frac{1}{e})\NBR\ge(1-\frac{1}{e})\ADAPT$, thus completing the proof.
\end{proof}

A similar analysis gives the following result:

\begin{theorem}
For the Stochastic Max Coverage problem, there exists a non-adaptive solution that achieves coverage $\tau(1-\frac{1}{e})\OFFLINE$.
\end{theorem}

\begin{proof}
The proof is similar to the proof of Theorem~\ref{thm:uniformnbr}. Again we use variables $x_{i,a}$ such that $x_{i,a}=1$ if and only if $S_{i,a}$ is included in the optimal subcollection $Q_a$ for $a$, and $y_{j,a}$ such that $y_{j,a}=1$ if and only if $e_j\in\bigcup_{i\in Q_a} S_{i,a}$. Clearly $\sum_{i:e_j\in S_{i,a}} x_{i,a} \ge y_{j,a}$ for each assignment $a$ and $\OFFLINE=\sum_j \E_a[y_{j,a}]$.

Again we randomly pick $k$ sets according to the following probability distribution: pick $i$ with probability $\frac{1}{k}\sum_a \Pr[a]\sum_{i=1}^n x_{i,a}$.

In expectation, the probability that picking a random $i$ produces a promising cover for $e_j$ is
\[ \frac{1}{k}\sum_a \Pr[a]\sum_{i:e_j\in S_{i,a}} x_{i,a} \ge \frac{1}{k}\sum_a \Pr[a]y_{j,a}=\frac{\E_a[y_{j,a}]}{k} \]
hence the probability that at least one promising cover for $j$ is chosen is
\[ 1-\left(1-\frac{\E_a[y_{j,a}]}{k}\right)^k\ge\left(1-\frac{1}{e}\right)\E_a[y_{j,a}]. \]
Since each promising cover covers $j$ with probability $\ge\tau$, the probability that $j$ is actually covered is at least $\tau(1-\frac{1}{e})\E_a[y_{j,a}]$. The rest of the proof follows from analysis similar to the analysis used in the proof of Theorem~\ref{thm:uniformnbr}.
\end{proof}

\section*{Acknowledgements}

Patrick Lin was partially supported by NSF Grants 1217968 and 1319648.
Devorah Kletenik and Lisa Hellerstein were partially supported by NSF Grants 1217968 and 0917153.


\appendix

%
%
%
%


\bibliographystyle{plain}
\bibliography{stochastic-coverage}

\section{Appendix: Proof of Proposition~\ref{thm:adapttau}}

We construct a counterexample for which $\ADAPT \approx (\tau-\eps)\OFFLINE$.

\begin{example}
Consider an instance of the stochastic submodular coverage with a cardinality constraint problem as follows: say $\ell=2$ (so this is over binary states), let $t = \frac{1}{\tau}$ and set $n = 1 + t^2$. Let $p_{i,1} = \tau$ and $p_{i,0} = 1-\tau$ for all $i \in N$. Let $B=1$, so the problem is to maximize the expected utility of picking a single bit.

Consider the following stochastic monotone submodular function $g$ defined over $\states$: for all subcollections $Q \subseteq N$ such that $i \notin Q$, let
\[
g_{Q,a}(i) =	\begin{cases}
			1 & \text{if $i =1$}\\
			t-\eps & \text{if $i \in \{2, \ldots, n\}$ and $a_i=1$} \\
			0 & \text{if $i \in \{2, \ldots, n\}$ and $a_i=0$}
			\end{cases}
\]
where the monotonicity and submodularity of $g$ follow from the fact that $g$ is an additive utility function.

For $i \notin Q$, $\E[g_{Q,b}(i)] = 1$ if $i = 1 $ and $\E[g_{Q,b}(i)] = \frac{t-\eps}{t}$ if $i \in \{2,\ldots,n\}$. Due to the low probability of any variable from the second group having a value of 1, no adaptive tree outperforms the greedy choice of picking the first  variable. Thus $\ADAPT=1$.

The optimal offline procedure, on the other hand, will pick any variable of the second group that has a value of 1. With probability $1-(1- \frac{1}{t})^{t^2}$, at least one variable of the second group will evaluate to 1; hence, $\OFFLINE = \left(1-(1- \frac{1}{t})^{t^2}\right) \cdot (t-\eps) + (1- \frac{1}{t})^{t^2} \cdot 1$.

It follows that for sufficiently large $t$, $\ADAPT \approx (\tau - \eps)\OFFLINE$.
\qed
\end{example}

\end{document}